\documentclass[a4paper,onecolumn,11pt,accepted=2021-10-21]{quantumarticle}
\pdfoutput=1

\usepackage[utf8]{inputenc}
\usepackage[english]{babel}
\usepackage[T1]{fontenc}
\usepackage{hyperref}

\usepackage{tikz}
\usepackage{lipsum}

\usepackage{amsmath,amssymb,amsfonts,xspace,graphicx,relsize,bm,mathtools,xcolor}
\usepackage{amsthm}
\usepackage{mathrsfs}

\usepackage[numbers,sort&compress]{natbib}

\def\01{\{0,1\}}
\newcommand{\ceil}[1]{\lceil{#1}\rceil}

\newcommand{\eps}{\varepsilon}
\newcommand{\ket}[1]{|#1\rangle}

\newcommand{\ketbra}[2]{|#1\rangle\langle#2|}

\newcommand{\diag}{\mbox{\rm diag}}
\newcommand{\polylog}{\mbox{\rm polylog}}
\newcommand{\norm}[1]{\mbox{$\parallel{#1}\parallel$}}

\newcommand{\Cc}{{\mathcal C}} 
\newcommand{\Exp}{\mathbb{E}}

\newcommand{\Sh}{\ensuremath{\mathcal{S}}}

\newcommand{\Fspan}{\mathrm{Fspan}}
\newcommand{\Fdim}{\mathrm{Fdim}}

\newcommand{\T}{\ensuremath{\mathcal{T}}}
\newcommand{\V}{\ensuremath{\mathcal{V}}}

\newcommand{\F}{\ensuremath{\mathbb{F}}}
\newcommand{\R}{\ensuremath{\mathbb{R}}}

\newcommand{\Se}{\ensuremath{\mathcal{S}}}

\newtheorem{definition}{Definition}
\newtheorem{theorem}{Theorem}
\newtheorem{lemma}{Lemma}

\newtheorem{fact}{Fact}
\newtheorem{corollary}{Corollary}
\newtheorem{conjecture}{Conjecture}
\newtheorem{observation}{Observation}
\newtheorem{claim}{Claim}

\newcommand{\pmset}[1]{\{-1,1\}^{#1}} 
\newcommand{\AND}{\mbox{\rm AND}}
\newcommand{\ADV}{\mbox{\rm ADV}}
\def\01{\{0,1\}}

\DeclareMathOperator{\Add}{Add}
\newcommand{\supp}{\mathrm{supp}}
\newcommand{\sgn}{\mathrm{sign}}
\DeclareMathOperator{\sign}{sign}
\DeclareMathOperator{\spann}{span}

\begin{document}

\title{Two new results about quantum exact learning}

\author{Srinivasan Arunachalam}
\affiliation{IBM T. J. Watson Research Center}
\thanks{Work done while a Postdoc at Center for Theoretical Physics, MIT and PhD student at QuSoft, CWI, Amsterdam, the Netherlands. Supported by ERC Consolidator  Grant 615307 QPROGRESS and  MIT-IBM Watson AI Lab under the project {\it Machine Learning in Hilbert space}. {\tt Srinivasan.Arunachalam@ibm.com}}

\author{Sourav Chakraborty}
\affiliation{Indian Statistical Institute, Kolkata, India}
\thanks{Work done while on sabbatical at CWI, supported by ERC Consolidator Grant 615307 QPROGRESS. {\tt sourav@isical.ac.in}}

\author{Troy Lee}
\affiliation{Centre for Quantum Software and Information, University of Technology Sydney, Australia}
\thanks{Partially supported by the Australian Research Council (Grant No: DP200100950).  Part of this work was done while at the School for Physical and Mathematical Sciences, Nanyang Technological University and the Centre for Quantum Technologies, Singapore, supported by the Singapore National Research Foundation under NRF RF Award No. NRF-NRFF2013-13. {\tt troyjlee@gmail.com}}

\author{Manaswi Paraashar}
\affiliation{Indian Statistical Institute, Kolkata, India}
\thanks{{\tt manaswi.isi@gmail.com}}

\author{Ronald de Wolf}
\affiliation{QuSoft, CWI and University of Amsterdam, the Netherlands}
\thanks{Partially supported by ERC Consolidator Grant 615307-QPROGRESS (which ended February 2019), and by the Dutch Research Council (NWO) through Gravitation-grant Quantum Software Consortium 024.003.037 and through QuantERA project QuantAlgo 680-91-034. {\tt rdewolf@cwi.nl}
\newline
\newline
A conference version of this paper appeared in the proceedings of the 46th International Colloquium on Automata, Languages and Programming (ICALP 19), Leibniz International Proceedings in Informatics (LIPIcs) volume 132, pp.16:1-16:15, 2019.}





\maketitle

\begin{abstract}
  We present two new results about exact learning by quantum computers.
		First, we show how to exactly learn a $k$-Fourier-sparse $n$-bit Boolean function from $O(k^{1.5}(\log k)^2)$ uniform quantum examples for that function.
		This improves over the bound of $\widetilde{\Theta}(kn)$ uniformly random \emph{classical} examples (Haviv and Regev, CCC'15). Additionally, we provide a possible direction to improve our $\widetilde{O}(k^{1.5})$ upper bound by proving an improvement of Chang's lemma for $k$-Fourier-sparse Boolean functions. Second, we show that if a concept class $\Cc$ can be exactly learned using $Q$ quantum membership queries, then it can also be learned using $O\left(\frac{Q^2}{\log Q}\log|\Cc|\right)$ \emph{classical} membership queries. This improves the previous-best simulation result (Servedio and Gortler, SICOMP'04) by a $\log Q$-factor.
\end{abstract}
\section{Introduction}
	
	\subsection{Quantum learning theory}
	
	Both quantum computing and machine learning are hot topics at the moment, and their intersection has been receiving growing attention in recent years as well.
	On the one hand there are particular approaches that use quantum algorithms like Grover search~\cite{grover:search} and the Harrow-Hassidim-Lloyd linear-systems solver~\cite{hhl:lineq} to speed up learning algorithms for specific machine learning tasks (see~\cite{wittek:qml,schuldea:introqml,adcockea:qml,biamonteea:qml,briegel&dunjko:qml} for recent surveys of this line of work). On the other hand there have been a number of more general results about the sample and/or time complexity of learning various concept classes using a quantum computer (see~\cite{arunachalam:quantumlearningsurvey} for a survey). This paper presents two new results in the latter line of work. In both cases the goal is to \emph{exactly} learn an unknown target function with high probability; for the first result our access to the target function is through quantum examples for the function, and for the second result our access is through membership queries to the function.
	
	\subsection{Exact learning of sparse functions from uniform quantum examples}
	
	Let us first explain the setting of distribution-dependent learning from examples.
	Let $\Cc$ be a class of functions, a.k.a.~a \emph{concept class}. For concreteness assume they are $\pm 1$-valued functions on a domain of size~$N$; if $N=2^n$, then the domain may be identified with $\01^n$. Suppose $c\in\Cc$ is an unknown function (the \emph{target} function or concept) that we want to learn. A learning algorithm is given \emph{examples} of the form $(x,c(x))$, where $x$ is distributed according to some probability distribution $D$ on $[N]$. An $(\eps,\delta)$-learner for~$\Cc$ w.r.t.~$D$ is an algorithm that, for every possible target concept $c\in\Cc$, produces a hypothesis $h:[N]\to\pmset{}$ such that  with probability at least $1-\delta$ (over the randomness of the learner and the examples for the target concept~$c$), $h$'s generalization error is at most $\eps$,~i.e.,
	$$
	\Pr_{x\sim D}[c(x)\neq h(x)]\leq\eps,
	$$
	where $x\sim D$ means $x$ is sampled according to the distribution $D$. 
	In other words, from $D$-distributed examples the learner has to construct a hypothesis that mostly agrees with the target concept \emph{under the same $D$}.
	
	In the early days of quantum computing, Bshouty and Jackson~\cite{bshouty:quantumpac} generalized this learning setting by allowing coherent \emph{quantum} examples. A quantum example for concept~$c$ w.r.t.\ distribution~$D$, is the following $(\ceil{\log N}+1)$-qubit state:
	$$
	\sum_{x\in[N]}\sqrt{D(x)}\ket{x,c(x)}.
	$$
	Clearly such a quantum example is at least as useful as a classical example, because measuring this state yields a pair $(x,c(x))$ where $x\sim D$.
	Bshouty and Jackson gave examples of concept classes that can be learned more efficiently from quantum examples than from classical random examples under specific~$D$. In particular, they showed that the concept class of DNF-formulas can be learned in polynomial time from quantum examples under the \emph{uniform} distribution, something we do not know how to do classically (the best classical upper bound is quasi-polynomial time~\cite{verbeurgt:learningdnf}). The key to this improvement is the ability to obtain, from a uniform quantum example, a sample $S\sim\widehat{c}(S)^2$ distributed according to the squared \emph{Fourier coefficients} of~$c$.\footnote{Parseval's identity implies $\sum_{S\in\01^n}\widehat{c}(S)^2=1$, so this is indeed a probability distribution.}
	This \emph{Fourier sampling}, originally due to Bernstein and Vazirani~\cite{bernstein&vazirani:qcomplexity}, is very powerful. For example, if $\Cc$ is the class of $\mathbb{F}_2$-linear functions on $\01^n$, then the unknown target concept~$c$ is a character function $\chi_S(x)=(-1)^{x\cdot S}$ \footnote{The linear functions with domain $\{0,1\}^n$ and range $\{0,1\}$ are defined as $(S \cdot x) \bmod{2}$, for $S \subseteq[n]$. The definition of linear functions we give here are for functions with range $\{-1,1\}$ rather than $\{0,1\}$.}; its only non-zero Fourier coefficient is $\widehat{c}(S)$ hence one Fourier-sample gives us the unknown~$S$ with certainty. In contrast, learning linear functions from classical uniform examples requires $\Theta(n)$ examples.  Another example where Fourier sampling is proven powerful is in learning the class of $\ell$-juntas on~$n$ bits.\footnote{We say $f:\01^n\rightarrow \pmset{}$ is an \emph{$\ell$-junta} if there exists a set $S\subseteq [n]$ of size $|S|\leq \ell$ such that $f$ depends only on the variables whose indices are in~$S$.} 
	{At\i c\i} and Servedio~\cite{atici&servedio:testing} showed that $(\log n)$-juntas can be exactly learned by a quantum learner under the uniform distribution in time polynomial in~$n$. 
	Classically it is a long-standing open question if a similar result holds when the learner is given uniform classical examples (the best known algorithm runs in quasi-polynomial time~\cite{mos:learningjuntas}). These cases (and others surveyed in~\cite{arunachalam:quantumlearningsurvey}) show that uniform quantum examples (and in particular Fourier sampling) can be more useful than classical~examples.\footnote{This is not the case in Valiant's \emph{PAC-learning} model~\cite{valiant:paclearning} of distribution-independent learning. There we require the same learner to be an $(\eps,\delta)$-learner for $\Cc$ w.r.t.\ \emph{every} possible distribution~$D$. One can show in this model (and also in the broader model of \emph{agnostic} learning) that the quantum and classical sample complexities are equal up to a constant factor~\cite{arunachalam:optimalpaclearning}.}

	In this paper we consider the concept class of $n$-bit Boolean functions (with domain $\{0,1\}^n$ and range $\{-1,1\}$) that are \emph{$k$-sparse in the Fourier domain}: $\widehat{c}(S)\neq 0$ for at most $k$ different $S$'s. This is a natural generalization of the above-mentioned case of learning linear functions, which corresponds to $k=1$. It also generalizes the case of learning $\ell$-juntas on $n$ bits, which are functions of sparsity $k=2^\ell$. Variants of the class of $k$-Fourier-sparse functions have been well-studied in the area of \emph{sparse recovery}, where the goal is to recover a $k$-sparse vector $x\in \R^N$ given a low-dimensional linear sketch~$Ax$ for a so-called ``measurement matrix'' matrix $A\in \R^{m\times N}$. See~\cite{hassanieh:sparserecovery,indyk:sampleoptimal} for some upper bounds on the size of the measurement matrix that suffice for sparse recovery. Closer to the setting of this paper, there has also been extensive work on learning the concept class of $n$-bit \emph{real-valued} functions that are $k$-sparse in the Fourier domain. In this direction Cheraghchi et al.~\cite{cheraghchi:RIPlistdecoding} showed that $O(nk(\log k)^3)$ uniform examples suffice to learn this concept class, improving upon the works of Bourgain~\cite{bourgain:RIP}, Rudelson and Vershynin~\cite{rudelsyn:sparsereconstruction} and Cand\'es and Tao~\cite{candesandtao:signalrecovery}.

	
	In this paper we focus on \emph{exactly} learning the target concept from uniform examples, with high success probability. So $D(x)=1/2^n$ for all $x$, $\eps=0$, and $\delta=1/3$. Haviv and Regev~\cite{haviv:listdecoding} showed that for classical learners $O(nk\log k)$ uniform examples suffice to learn $k$-Fourier-sparse functions, and $\Omega(nk)$ uniform examples are necessary. In Section~\ref{sec:ksparse} we study the number of uniform \emph{quantum} examples needed to learn $k$-Fourier-sparse Boolean functions, and show that it is upper bounded by $O(k^{1.5}(\log k)^2)$. For $k\ll n^2$ this quantum bound is much better than the number of uniform examples used in the classical case. Proving the upper bound is done in two phases. In the first phase we use  the fact that a uniform quantum example allows us to Fourier-sample the target concept and, with some Fourier analysis of $k$-Fourier-sparse functions, we learn the Fourier span using $O(rk)$ examples, where $r$ is the Fourier dimension of the target concept (see Section~\ref{sec: Preliminaries} for the definition of Fourier dimension). %
	In the second phase, we reduce the number of variables to the dimension~$r$ of the Fourier support, and then invoke the classical learner of Haviv and Regev to learn the target function from $O(rk\log k)$ classical examples. Since it is known that $r=O(\sqrt{k}\log k)$~\cite{sanyal:fourierdim}, the two phases together imply that $O(k^{1.5}(\log k)^2)$ uniform quantum examples suffice to exactly learn the target with high probability. 
	We also prove a (non-matching) lower bound of $\Omega(k\log k)$ uniform quantum examples, using techniques from quantum information theory.

	We believe that the sample complexity for Phase~1 of our learning algorithm is actually $\tilde{O}(k)$.
	Towards that end, we propose a possible way to prove the sample complexity of our Phase~1 to $\Tilde{O}(k)$. The first step in Phase~1 of our algorithm is to obtain an $S \neq 0^n$ such that $\widehat{c}(S) \neq 0$, where $c$ is the $k$-Fourier-sparse target concept.  It follows from Chang's lemma~\cite{chang:inequality},
	a central result in additive combinatorics, that in expectation $O(k \sqrt{\log k}/\sqrt{r})$ Fourier-samples are sufficient to obtain one such $S$. In Section~\ref{sec: potential direction to improve phase 1} we present an improvement of Chang's lemma for the case of $k$-Fourier-sparse Boolean functions. Using this improvement we can show that in expectation $O((k \log k)/r)$ Fourier-samples are sufficient to obtain an $S\neq \emptyset$ such that $\widehat{c}(S) \neq 0$. We conjecture (Conjecture~\ref{conjecture : phase1}) a generalization of our improvement of Chang's Lemma which, if true, would imply that Phase~1 of our algorithm can be done in $\Tilde{O}(k)$ many expected number of samples. Our improvement of Chang's lemma and the techniques used therein might be of independent interest.
	


	\subsection{Exact learning from quantum membership queries}
	Our second result is in a model of active learning. The learner still wants to exactly learn an unknown target concept $c:[N]\to\pmset{}$ from a known concept class~$\Cc$, but now the learner can choose which points of the truth-table of the target it sees, rather than those points being chosen randomly. More precisely, the learner can query $c(x)$ for any $x$ of its choice. This is called a \emph{membership query}.\footnote{Think of the set $\{x\mid c(x)=1\}$ corresponding to the target concept: a membership query asks whether $x$ is a member of this set or not.}
	Quantum algorithms have the following query operation available:
	$$
	O_c:\ket{x,b}\mapsto\ket{x,b\cdot c(x)},
	$$
	where $b\in\pmset{}$. For some concept classes, quantum membership queries can be much more useful than classical. Consider again the class $\Cc$ of $\mathbb{F}_2$-linear functions on $\01^n$. Using one query to a uniform superposition over all $x$ and doing a Hadamard transform, we can Fourier-sample and hence learn the target concept exactly. In contrast, $\Theta(n)$ classical membership queries are necessary and sufficient for classical learners.
	As another example, consider the concept class $\Cc=\{\delta_i\mid i\in[N]\}$ of the $N$ point functions, where $\delta_i(x)=1$ iff $i=x$. Elements from this class can be learned using $O(\sqrt{N})$ quantum membership queries by Grover's algorithm, while every classical algorithm needs to make $\Omega(N)$ membership~queries.
	
	For a given concept class $\Cc$ of $\pm 1$-valued function on $[N]$, let $D(\Cc)$ denote the minimal number of classical membership queries needed for learners that can exactly identify every $c\in\Cc$ with success probability~1 (such learners are deterministic without loss of generality). Let $R(\Cc)$ and $Q(\Cc)$ denote the minimal number of classical and quantum membership queries, respectively, needed for learners that can exactly identify every $c\in\Cc$ with error probability $\leq 1/3$.\footnote{We can identify each concept with a string $c\in\pmset{N}$, and hence $\Cc\subseteq\pmset{N}$. The goal is to learn the unknown $c\in\Cc$ with high probability using few queries to the corresponding $N$-bit string. This setting is also sometimes called ``oracle identification'' in the literature; see \cite[Section~4.1]{arunachalam:quantumlearningsurvey} for more references.} Servedio and Gortler~\cite{servedio&gortler:equivalencequantumclassical} showed that these quantum and classical measures cannot be too far apart. First, using an information-theoretic argument they~showed
	$$
	Q(\Cc)\geq \Omega\left(\frac{\log|\Cc|}{\log N}\right).
	$$
	Intuitively, this holds because a learner recovers roughly $\log|\Cc|$ bits of information, while every quantum membership query can give at most $O(\log N)$ bits of information.
	Note that this is tight for the class of linear functions, where the left- and right-hand sides are both constant. Second, using the so-called hybrid method they showed
	$$
	Q(\Cc)\geq \Omega(1/\sqrt{\gamma(\Cc)}),
	$$
	for some combinatorial parameter~$\gamma(\Cc)$ that we will not define here (but which is $1/N$ for the class~$\Cc$ of point functions, hence this inequality is tight for that $\Cc$).
	They also noted the following upper bound:
	$$
	D(\Cc)=O\left(\frac{\log|\Cc|}{\gamma(\Cc)}\right).
	$$
	Combining these three inequalities yields the following relation between $D(\Cc)$ and $Q(\Cc)$
	\begin{equation}\label{eq:DQrelation}
	D(\Cc)\leq O(Q(\Cc)^2\log|\Cc|)\leq O(Q(\Cc)^3\log N).
	\end{equation}
	This shows that, up to a $\log N$-factor, quantum and classical membership query complexities of exact learning are polynomially close.
	While each of the three inequalities that together imply \eqref{eq:DQrelation} can be individually tight (for different $\Cc$), this does not imply \eqref{eq:DQrelation} itself is tight. 
	
	Note that Eq.~\eqref{eq:DQrelation} upper bounds the membership query complexity of \emph{deterministic} classical learners. We are not aware of a stronger upper bound on \emph{bounded-error} classical learners.
	However, in Section~\ref{sec:QvsD} we tighten that bound further by a $\log Q(\Cc)$-factor: 
	$$
	R(\Cc)\leq O\left(\frac{Q(\Cc)^2}{\log Q(\Cc)}\log|\Cc|\right)\leq O\left(\frac{Q(\Cc)^3}{\log Q(\Cc)}\log N\right).
	$$
	This inequality is tight both for the class of linear functions and the class of point functions.
	
	Our proof combines the quantum adversary method~\cite{ambainis:lowerboundsj,bss:semidef,spalek&szegedy:adversary} with an entropic argument to show that we can always find a query whose outcome (no matter whether it is $1$ or $-1$) will shrink the concept class by a factor $\leq 1-\frac{\log Q(\Cc)}{Q(\Cc)^2}$. While our improvement over the earlier bounds is not very large, we feel our usage of entropy to save a log-factor is new and may have applications elsewhere.

	\section{Preliminaries}
	\label{sec: Preliminaries}
	
	\paragraph{Notation.} Let $[n]=\{1,\ldots,n\}$. For an $n$-dimensional vector space, the standard basis vectors are $\{e_i\in\01^n \mid i\in [n]\}$, where $e_i$ is the vector with a $1$ in the $i$th coordinate and zeros elsewhere.  
	For $x\in \01^n$ and $i\in [n]$, let $x^i$ be the input obtained by flipping the $i$th bit in $x$.
	
	For a Boolean function $f:\01^n\rightarrow \pmset{}$ and $B\in \F_2^{n\times n}$, define $f\circ B:\01^n\rightarrow \pmset{}$ as $(f\circ B)(x):=f(Bx)$, where the matrix-vector product $Bx$ is over $\F_2$. Throughout this paper, the rank of a matrix $B\in \F_2^{n\times n}$ will be taken over $\F_2$. Let $B_1,\ldots,B_n$ be the columns of $B$. 
	
	\paragraph{Fourier analysis on the Boolean cube.} We introduce the basics of Fourier analysis here, referring to~\cite{odonnell:analysis,wolf:fouriersurvey} for more. Define the inner product between functions $f,g:\01^n\rightarrow \R$~as
	$$
	\langle f,g\rangle=\Exp_{x\in \01^n} [f(x)\cdot g(x)],
	$$ 
	where the expectation is uniform over all $x\in \01^n$. For $S\in\01^n$, the character function corresponding to $S$ is given by $\chi_S(x):=(-1)^{S\cdot x}$, where the dot product $S\cdot x$ is $\sum_{i=1}^n S_ix_i$. For every $j \in [n]$, we use the notation $\chi_j$ to denote the function $\chi_{\{j\}}$. Observe that the set of functions $\{\chi_S\}_{S\in \01^n}$ forms an orthonormal basis for the space of real-valued functions over the Boolean cube. Hence every $f:\01^n\rightarrow \R$ can be written uniquely as 
	$$
	f(x)=\sum_{S\in \01^n} \widehat{f}(S) (-1)^{S\cdot x} \quad \text{for all }x\in \01^n,
	$$ 
	where $\widehat{f}(S)=\langle f,\chi_S\rangle=\Exp_x[f(x)\chi_S(x)]$ is called a \emph{Fourier coefficient} of $f$. For $i\in [n]$, we write $\widehat{f}(e_i)$ as $\widehat{f}(i)$ for notational convenience. 
	
	Parseval's identity states that $\sum_{S\in \01^n}\widehat{f}(S)^2=\Exp_x[f(x)^2]$. If $f$ has range $\pmset{}$, then Parseval gives $\sum_{S\in \01^n}\widehat{f}(S)^2=1,$ so $\{\widehat{f}(S)^2\}_{S\in \01^n}$ forms a probability distribution. 
	The \emph{Fourier weight} of function~$f$ on $\Sh\subseteq \01^n$ is defined as $\sum_{S\in \Sh} \widehat{f}(S)^2$.
	
	For $f:\01^n\rightarrow \R$, the \emph{Fourier support} of $f$ is $\supp(\widehat{f})=\{S:\widehat{f}(S)\neq 0\}$. The \emph{Fourier sparsity} of $f$ is $|\supp(\widehat{f})|$. The \emph{Fourier span} of $f$, denoted $\Fspan(f)$, is the span of $\supp(\widehat{f})$. The \emph{Fourier dimension} of $f$, denoted $\Fdim(f)$, is the dimension of the Fourier span. We say $f$ is \emph{$k$-Fourier-sparse} if $|\supp(\widehat{f})|\leq k$.

	We now state a number of known structural results about Fourier coefficients and dimension.
	
	\begin{theorem}[\cite{sanyal:fourierdim}]
	\label{thm:fourierdimensionality}
		The Fourier dimension of a $k$-Fourier-sparse $f:\01^n\rightarrow \pmset{}$ is $O(\sqrt{k}\log k)$.\footnote{Note that this theorem is optimal up to the logarithmic factor for the addressing function $\Add_m:\01^{\log m+m}\rightarrow \pmset{}$ defined as $\Add_m(x,y)=1-2y_x$ for all $x\in \01^{\log m}$ and $y\in \01^m$, i.e., the output of $\Add_m(x,y)$ is determined by the value $y_x$, where $x$ is treated as the binary representation of a number in $\{0,\ldots,m-1\}$. For the $\Add_m$ function, the Fourier dimension is $m$ and the Fourier sparsity is~$m^2$.}
	\end{theorem}

	\begin{lemma}[{\cite[Theorem~12]{gopalan:fouriersparsity}}]
		\label{lemma:granularityofsparse}
		Let $k\geq 2$. The Fourier coefficients of a $k$-Fourier-sparse Boolean function $f:\01^n\rightarrow \pmset{}$ are integer multiples of $2^{1-\lfloor \log k\rfloor}$.
	\end{lemma}
	
	\begin{definition}
		\label{defn:fourierbasis}
		Let $f:\01^n \to \pmset{}$  and suppose $B\in \F_2^{n\times n}$ is invertible. Define~$f_{B}$~as 
		$$
		f_{B}(x) = f((B^{-1})^{\mathsf T}x).
		$$ 
	\end{definition}

	\begin{lemma}\label{lemma:Fouriercoeffmatrixprod}
		Let $f:\01^n\rightarrow \R$ and suppose $B\in \F_2^{n\times n}$ is invertible. Then the Fourier coefficients of $f_B$ are $\widehat{f_B}(Q)=\widehat{f}(BQ)$ for all $Q\in \01^n$.  
	\end{lemma}
	
	\begin{proof}
		Write out the Fourier expansion of $f_B$:
		\begin{align*}
		    f_B(x) 
		    = f((B^{-1})^\mathsf{T}x) 
		    =\sum_{S\in \01^n}\widehat{f}(S)(-1)^{(B^{-1}S)\cdot x} 
		    =\sum_{Q\in \01^n}\widehat{f}(BQ)(-1)^{Q\cdot x},
		\end{align*}
		where the second equality used $\langle S,(B^{-1})^\mathsf{T}x\rangle =\langle B^{-1} S,x\rangle$ and the last used the substitution $S=BQ$.
	\end{proof}
	
	The following lemma (Lemma~\ref{lemma:basis}) easily follows by applying Lemma~\ref{lemma:Fouriercoeffmatrixprod} with an invertible linear map $B$ that maps $e_i$ to $B_i$, for every $i \in [r]$.
	
	\begin{lemma}\label{lemma:basis} 
		Let $f:\01^n\rightarrow \pmset{}$, and $B\in \F_2^{n\times n}$ be an invertible matrix such that the first $r$ columns of~$B$ are a 
		basis of the Fourier span of $f$, and $\widehat{f}(B_1), \ldots, \widehat{f}(B_r)$ are non-zero. Then 
		\begin{enumerate} 
			\item The Fourier span of $\widehat{f_{B}}$ is spanned by $\{e_1, \dots, e_r\}$, i.e., $f_{B}$ has only $r$ influential variables. 
			\item For every $i \in [r]$, $\widehat{f_{B}}(i) \neq 0$. 
		\end{enumerate}
	\end{lemma}
	
	Here is the well-known fact, already mentioned in the introduction, that one can Fourier-sample from uniform quantum examples:
	
	\begin{lemma}\label{lemma:fouriersampleusingexamples}
		Let $f:\01^{n}\rightarrow \pmset{}$. There exists a procedure that uses one uniform quantum example and satisfies the following: with probability $1/2$ it outputs an $S$ drawn from the distribution $\{\widehat{f}(S)^2\}_{S\in \01^n}$, otherwise it rejects.
	\end{lemma}
	
	\begin{proof}
		Using a uniform quantum example $\frac{1}{\sqrt{2^n}}\sum_x \ket{x,f(x)}$, one can obtain $\frac{1}{\sqrt{2^n}}\sum_x f(x) \ket{x}$ with probability~$1/2$: replace $f(x)\in\pmset{}$ by $(1-f(x))/2\in\01$ unitarily, apply the Hadamard transform to the last qubit and measure it. With probability $1/2$ we obtain the outcome~0, in which case our procedure rejects. Otherwise the remaining state is $\frac{1}{\sqrt{2^n}}\sum_x f(x) \ket{x}$. Apply Hadamard transforms to all $n$ qubits to obtain $\sum_S \widehat{f}(S)\ket{S}$. Measuring this quantum state gives an $S$ with probability~$\widehat{f}(S)^2$. 
	\end{proof}
	
	\paragraph{Information theory.}
	We refer to~\cite{cover&thomas:infoth} for a comprehensive introduction to classical information theory, and here just remind the reader of the basic definitions. A random variable $\mathbf{A}$ with probabilities $\Pr[\mathbf{A}=a]=p_a$ has \emph{entropy} $H(\mathbf{A}):=-\sum_a p_a\log(p_a)$. For a pair of (possibly correlated) random variables $\mathbf{A},\mathbf{B}$, the \emph{conditional entropy} of $\mathbf{A}$ given $\mathbf{B}$, is $H(\mathbf{A}\mid \mathbf{B}):=H(\mathbf{A},\mathbf{B})-H(\mathbf{B})$. This equals $\Exp_{b\sim \mathbf{B}}[H(\mathbf{A}\mid \mathbf{B}=b)]$.
	The \emph{mutual information} between $\mathbf{A}$ and $\mathbf{B}$ is $I(\mathbf{A}:\mathbf{B}):=H(\mathbf{A})+H(\mathbf{B})-H(\mathbf{A},\mathbf{B})=H(\mathbf{A})-H(\mathbf{A}\mid \mathbf{B})$.
	The \emph{binary entropy} $H(p)$ is the entropy of a bit with distribution $(p,1-p)$. If $\rho$ is a density matrix (i.e., a trace-1 positive semi-definite matrix), then its singular values form a probability distribution~$P$, and the \emph{von Neumann entropy} of $\rho$ is~$S(\rho):=H(P)$.
	We refer to~\cite[Part~III]{nielsen&chuang:qc} for a more extensive introduction to quantum information theory.

	\section{Exact learning of $k$-Fourier-sparse functions}\label{sec:ksparse}
	
	In this section we consider exactly learning the concept class~$\Cc$ of $k$-Fourier-sparse Boolean functions: 
	$$
	\Cc=\{f:\01^n\rightarrow \pmset{}:|\supp(\widehat{f})|\leq k\}.
	$$   
	The goal is to exactly learn $c\in \Cc$ given \emph{uniform examples} from $c$ of the form $(x,c(x))$  where $x$ is drawn from the uniform distribution on $\01^n$. Haviv and Regev~\cite{haviv:listdecoding} considered learning this concept class and showed the following results.
	
	
	\begin{theorem}[Corollary~3.6 of \cite{haviv:listdecoding}]\label{thm:havivregevub}
		For every $n>0$ and $k\leq 2^n$, the number of uniform examples that suffice to learn $\Cc$ with probability $1-2^{-\Omega(n\log k)}$ is $O(nk\log k)$.
	\end{theorem}
	
	\begin{theorem}[Theorem~3.7 of \cite{haviv:listdecoding}]
		For every $n>0$ and $k\leq 2^n$, the number of uniform examples necessary to learn $\Cc$ with constant success probability is $\Omega(k(n-\log k))$.
	\end{theorem}
	
	Our main results in this section are about the number of uniform \emph{quantum} examples that are necessary and sufficient to exactly learn the class $\Cc$ of $k$-Fourier-sparse functions. A uniform quantum example for a concept $c\in \Cc$ is the quantum state  
	$$
	\frac{1}{\sqrt{2^n}}\sum_{x\in \01^n}\ket{x,c(x)}.
	$$
	
	Our first theorem of this section (Section~\ref{sec: Upper bound on learning k-Fourier-sparse Boolean functions}) gives an upper bound on the number of uniform quantum examples that are sufficient to learn $\mathcal{C}$ by giving a learning algorithm.
	\begin{theorem}\label{thm:ourresultupperbound}
		For every $n>0$ and $k\leq 2^n$, the number of uniform quantum examples that suffice to~learn~$\Cc$ with probability $\geq 2/3$ is $O(k^{1.5}(\log k)^2)$.
	\end{theorem}
	The learning algorithm has two phases: Phase~1 is described in Section~\ref{sec: Phase 1: Learning the Fourier span} and Phase~2 is discussed in Section~\ref{sec: Phase 2: Learning the function completely}.

	In the theorem below (Section~\ref{sec: Lower bound  on learning k-Fourier-sparse Boolean functions}) we prove the following (non-matching) lower bound on the number of uniform quantum examples necessary to learn~$\Cc$.
	
	\begin{theorem}
	\label{thm:ourresult lower bound}
		For every $n>0$, constant $c\in (0,1)$ and $k\leq 2^{cn}$, the number of uniform quantum examples necessary to learn~$\Cc$ with constant success probability is $\Omega(k\log k)$.
	\end{theorem}
	
	In Section~\ref{sec: potential direction to improve phase 1} we give a possible direction to prove an improved sample complexity for Phase~1 of our learning algorithm.

	\subsection{Upper bound on learning $k$-Fourier-sparse Boolean functions}
	\label{sec: Upper bound on learning k-Fourier-sparse Boolean functions}
We split our quantum learning algorithm into two phases. Suppose $c\in \Cc$ is the unknown concept, with Fourier dimension~$r$. In the first phase the learner uses samples from the distribution $\{\widehat{c}(S)^2\}_{S\in \01^n}$ to learn the Fourier span of $c$. In the second phase the learner uses uniform \emph{classical} examples to learn $c$ exactly, knowing its Fourier span. Phase~1 uses $O(rk)$ uniform quantum examples (for Fourier-sampling) and Phase~2 uses $O(rk\log k)$ uniform \emph{classical} examples. 
	
	\begin{theorem} 
		\label{thm:quantumlearnerforr}
		Let $k,r>0$. There exists a quantum learner that exactly learns (with high probability) an unknown $k$-Fourier-sparse $c:\01^n \to \pmset{}$ with Fourier dimension upper bounded by some known~$r$, from $O(rk\log k)$ uniform quantum examples.
	\end{theorem}
	
	The learner may not know the exact Fourier dimension~$r$ in advance, but Theorem~\ref{thm:fourierdimensionality} gives an upper bound $r=O(\sqrt{k}\log k)$, so our Theorem~\ref{thm:ourresultupperbound} follows immediately from     Theorem~\ref{thm:quantumlearnerforr}.

	 Before we prove this Theorem~\ref{thm:quantumlearnerforr}, we first  give a ``trivial'' algorithm for learning the Fourier support of Fourier-sparse functions quantumly. Gopalan et al.~\cite{gopalan:fouriersparsity} showed that every $k$-Fourier-sparse Boolean function is ``$2^{-\lceil \log k\rceil}$-granular'', i.e., every Fourier coefficient of a $k$-Fourier-sparse Boolean function $c$ is either $0$ or an integer multiple of $2^{-\lceil \log k\rceil}$. Using this observation, if one is allowed to Fourier-sample from $c$, then each $S$ with non-zero $\widehat{c}(S)$ will be observed with probability $\Omega(1/k^2)$, and using a coupon collector argument, we obtain the entire Fourier support using $O(k^2\log k)$ many Fourier-samples. Our main contribution in Theorem~\ref{thm:quantumlearnerforr} is to use the Fourier \emph{dimension} in order to improve this trivial quantum algorithm. In particular observe that for functions with Fourier dimension $\log k$ (such as $(\log k)$-juntas), the theorem above scales as $O(k\log^2k)$ which is better than the trivial algorithm by a factor of nearly $k$.
	
	
\subsubsection{Phase~1: Learning the Fourier span}
\label{sec: Phase 1: Learning the Fourier span}

In this phase of the algorithm our goal is to learn the $r$-dimensional Fourier span of the $k$-Fourier-sparse target concept~$c$, using $O(r k)$ Fourier-samples. The algorithm is very simple: Fourier-sample more and more $S$'s and keep track of their span; stop when we reach dimension~$r$. The key is the following technical lemma, which says that if our current span $V'$ does not yet equal the full Fourier span~$V$, then there is significant Fourier weight outside of $V'$. This implies that a small expected number of additional Fourier-samples will give us an $S\in V\setminus V'$, which will grow our current span. After $r$ such grow-steps we have learned the full Fourier span.
	
	\begin{lemma}
	\label{lem:phase1}
Let $V\subseteq\01^n$ be the $r$-dimensional Fourier span of $k$-Fourier-sparse function $c:\01^n\to\pmset{}$, and $V'\subseteq V$ be a proper subspace. Then $\sum_{S\in V\setminus V'}\widehat{c}(S)^2\geq 1/k$.
\end{lemma}

\begin{proof}
Let us assume the worst case, which is that $\dim(V')=r-1$.
Because we can do an invertible linear transformation on $c$ as in Lemma~\ref{lemma:Fouriercoeffmatrixprod}, we may assume without loss of generality that the one ``missing'' dimension corresponds to the variable $x_r$ (i.e., $V=\spann(V'\cup\{e_r\})$).
Let $g$ be the (not necessarily Boolean-valued) part of $c$ with Fourier coefficients in $V'$:
$$
g(x):=\sum_{S\in V'}\widehat{c}(S)\chi_S(x).
$$
Suppose, towards a contradiction, that the Fourier weight $W:=\sum_{S\in V\setminus V'}\widehat{c}(S)^2$ is $<1/k$. This implies that $c$ and $g$ have the same sign on every $x\in\01^n$, as follows (using Cauchy-Schwarz):
$$
|c(x)-g(x)|=\left|\sum_{S\in V\setminus V'}\widehat{c}(S)\chi_S(x)\right|\leq\sqrt{kW}<1.
$$
Since $c$ depends on the variable $x_r$, there exists an $x\in\01^n$ where $x_r$ is influential, i.e., $c(x)\neq c(x^r)$. But $g$ is independent of $x_r$, which implies $c(x) = \sgn(g(x)) = \sgn(g(x^r)) = c(x^r)$, a contradiction. Hence $W\geq 1/k$.
\end{proof}

We now conclude Phase~1 by presenting a quantum learning algorithm that learns the Fourier span of an unknown $r$-dimensional $c\in \Cc$, given uniform quantum examples for~$c$. 

\begin{theorem}
\label{theo: phase1_main_thm}
Let $k,r>0$. 
There exists a quantum learner that uses uniform quantum examples for an unknown $k$-Fourier-sparse $c:\01^{n} \to \pmset{}$ with Fourier dimension~$r$. After processing each new quantum example it outputs a subspace of the Fourier span of~$c$. This sequence of subspaces is non-decreasing, and after an expected number of at most $2rk$ quantum examples, the output equals the Fourier span of $c$.  
\end{theorem}

This quantum learner can actually run forever, but if we know the Fourier dimension $r$ of~$c$, or an upper bound~$r$ on the actual Fourier dimension (e.g., by Theorem~\ref{thm:fourierdimensionality}), then we can stop the learner after processing $6rk$ examples; now, by Markov's inequality, with probability $\geq 2/3$ the last subspace will be the Fourier span of~$c$. 

\medskip

\begin{proof}
In order to learn the Fourier span of $c$, the quantum learner simply takes Fourier-samples until they span an $r$-dimensional space.
 Since we can generate a Fourier-sample from an expected number of 2 uniform quantum examples (by Lemma~\ref{lemma:fouriersampleusingexamples}), the expected number of uniform quantum examples needed is at most twice the expected number of Fourier-samples. If our current sequence of Fourier-samples spans an $r'$-dimensional space $V'$, with $r'<r$, then Lemma~\ref{lem:phase1} implies that the next Fourier-sample has probability at least $1/k$ of yielding an $S\not\in V'$. Hence an expected number of at most $k$ Fourier-samples suffices to grow the dimension of $V'$ by at least~1. Since we stop at dimension~$r$, the overall expected number of Fourier-samples is at most~$2rk$.
\end{proof}

\subsubsection{Phase~2: Learning the function completely}
\label{sec: Phase 2: Learning the function completely}
	
	In the above Phase~1, the quantum learner obtains the Fourier span of~$c$, which we will denote by~$\T$. Using this, the learner can restrict to the following concept class
	$$
	\Cc'=\{c:\01^{n}\rightarrow \pmset{} \mid c \text{ is } k\text{-Fourier-sparse with Fourier span }\T\}
	$$
	Let $\dim(\T)=r$. Let $B\in \F_2^{n\times n}$ be an invertible matrix whose first $r$ columns form a basis for $\T$. Consider $c_B=c\circ (B^{-1})^\mathsf T$ for $c\in \Cc'$. By Lemma~\ref{lemma:basis} it follows that $c_B$ depends on only its first~$r$ bits, and we can write $c_B:\01^r\rightarrow \pmset{}$. Hence the learner can apply the transformation $c\mapsto c\circ (B^{-1})^{\mathsf T}$ for every $c\in \Cc'$ and restrict to the concept class
	$$
	\Cc_r'=\{c':\01^{r}\rightarrow \pmset{} \mid c'=c\circ(B^{-1})^\mathsf{T}\text{ for some }c\in\Cc'\text{ and invertible }B\}.
	$$
	We now conclude Phase~2 of the algorithm by invoking the classical upper bound of Haviv-Regev (Theorem~\ref{thm:havivregevub}) which says that $O(rk\log k)$ uniform classical examples of the form $(z,c'(z))\in\01^{r+1}$ suffice to learn $\Cc'_r$. Although we assume our learning algorithm has access to uniform examples of the form $(x,c(x))$ for $x\in \01^n$, the quantum learner knows $B$ and hence can obtain a uniform example $(z,c'(z))$ for $c'$ by letting $z$ be the first $r$ bits of $B^\mathsf{T}x$ and $c'(z)=c(x)$.

\subsection{Lower bound  on learning $k$-Fourier-sparse Boolean functions}
\label{sec: Lower bound  on learning k-Fourier-sparse Boolean functions}
	
	In this section we show that $\Omega(k\log k)$ uniform quantum examples are necessary to learn the concept class of $k$-Fourier-sparse Boolean functions. 
	
	\begin{theorem}\label{thm:quantumlooselowerbound}
		For every $n$, constant $c\in (0,1)$ and $k\leq 2^{cn}$, the number of uniform quantum examples necessary to learn the class of $k$-Fourier-sparse Boolean functions, with success probability $\geq 2/3$, is~$\Omega(k\log k)$.
	\end{theorem}
	
	\begin{proof}
		Assume for simplicity that $k$ is a power of~2, so $\log k$ is an integer.
		We prove the lower bound for the following concept class, which was also used for the classical lower bound of Haviv and Regev~\cite{haviv:listdecoding}: let $\V$ be the set of distinct subspaces in $\01^{n}$ with dimension $n-\log k$ and
		$$
		\Cc=\{c_V:\01^n\rightarrow \pmset{} \mid c_V(x)=-1 \text{ iff } x\in V,\text{ where } V\in \V\}.
		$$
		Note that 
		every function in $\Cc$ has Fourier sparsity at most $k$,
		$|\Cc|=|\V|$, and each $c_V\in\Cc$ evaluates to 1 on a $(1-1/k)$-fraction of its domain.
		
		We prove the lower bound for $\Cc$ using a three-step information-theoretic technique. A similar approach was used in proving classical and quantum PAC learning lower bounds in~\cite{arunachalam:optimalpaclearning}.  Let $\mathbf{A}$ be a random variable that is uniformly distributed over $\Cc$. Suppose $\mathbf{A}=c_V$, and let $\mathbf{B}=\mathbf{B}_1\ldots\mathbf{B}_T$ be $T$ copies of the quantum example $$\ket{\psi_V}=\frac{1}{\sqrt{2^n}}\sum_{x\in \01^{n}}\ket{x,c_V(x)}$$ for $c_V$. The random variable $\mathbf{B}$ is a function of the random variable~$\mathbf{A}$. 
		The following upper and lower bounds on $I(\mathbf{A}:\mathbf{B})$ are similar to~\cite[proof of Theorem~12]{arunachalam:optimalpaclearning} and we omit the details of the first two steps here.
		\begin{enumerate}
			\item $I(\mathbf{A}:\mathbf{B})\geq \Omega(\log|\V|)$ because $\mathbf{B}$ allows one to recover $\mathbf{A}$ with high probability.
			\item $I(\mathbf{A}:\mathbf{B})\leq T\cdot I(\mathbf{A}:\mathbf{B}_1)$ using a chain rule for mutual information.
			
			\item $I(\mathbf{A}:\mathbf{B}_1)\leq O(n/k)$.\\[1mm]
			\emph{Proof (of 3).} Since $\mathbf{A}\mathbf{B}$ is a classical-quantum state, we have 
			$$
			I(\mathbf{A}:\mathbf{B}_1)= S(\mathbf{A})+S(\mathbf{B}_1)-S(\mathbf{A}\mathbf{B}_1)=S(\mathbf{B}_1),
			$$ 
			where the first equality is by definition and the second equality uses $S(\mathbf{A})=\log |\V|$ since $\mathbf{A}$ is uniformly distributed over~$\Cc$, and $S(\mathbf{A}\mathbf{B}_1)=\log |\V|$ since the matrix 
			$$
			\sigma=\frac{1}{|\V|} \sum_{V\in \V} \ketbra{V}{V}\otimes \ketbra{\psi_V}{\psi_V}
			$$ is block-diagonal with $|\V|$ rank-1 blocks on the diagonal. It thus suffices to bound the entropy of the (vector of singular values of the) reduced state of $\mathbf{B}_1$, which~is
			$$
			\rho=\frac{1}{|\V|}\sum_{V\in \V}\ketbra{\psi_V}{\psi_V}.
			$$
			Let $\sigma_0\geq \sigma_1\geq\cdots\geq \sigma_{2^{n+1}-1}\geq 0$ be the singular values of $\rho$. Since~$\rho$ is a density matrix, these form a probability distribution. 
			Now observe that $\sigma_0\geq 1-1/k$ since the inner product between $\frac{1}{\sqrt{2^n}}\sum_{x\in \01^n}\ket{x,1}$ and every $\ket{\psi_V}$ is $1-1/k$. 
			Let $\mathbf{N}\in\{0,1,\ldots,2^{n+1}-1\}$ be a random variable with probabilities $\sigma_0,\sigma_1,\ldots,\sigma_{2^{n+1}-1}$, and $\mathbf{Z}$ an indicator for the event ``$\mathbf{N}\neq 0$.'' Note that $\mathbf{Z}=0$ with probability $\sigma_0\geq 1-1/k$, and $H(\mathbf{N}\mid \mathbf{Z}=0)=0$. By a similar argument as in~\cite[Theorem~15]{arunachalam:optimalpaclearning}, we~have 
			\begin{align*}
			S(\rho) & =H(\mathbf{N})=H(\mathbf{N},\mathbf{Z})=H(\mathbf{Z})+H(\mathbf{N}\mid\mathbf{Z})\\
			& =H(\sigma_0)+\sigma_0\cdot H(\mathbf{N}\mid \mathbf{Z}=0) + (1-\sigma_0)\cdot H(\mathbf{N}\mid \mathbf{Z}=1) \\
			& \leq H\Big(\frac{1}{k}\Big) + \frac{n+1}{k}\leq O\Big(\frac{n+\log k}{k}\Big) 
			\end{align*}
			using $H(\alpha)\leq O(\alpha\log (1/\alpha))$.
		\end{enumerate}
		Combining these three steps implies $T=\Omega(k(\log |\V|)/n)$.
		It remains to lower bound $|\V|$. 
		
		\begin{claim}
			The number of distinct $d$-dimensional subspaces of $\F_2^n$ is at least $ 2^{\Omega((n-d)d)}$.
		\end{claim}
		
		\begin{proof}
			We can specify a $d$-dimensional subspace by giving $d$ linearly independent vectors in it. The number of distinct sequences of $d$ linearly independent vectors is exactly $(2^n-1)(2^n-2)(2^n-4)\cdots (2^n-2^{d-1})$, because once we have the first $t$ linearly independent vectors, with span $\Se_t$, then there are $2^n-2^t$ vectors that do not lie in $\Se_t$. 
			
			However, we are double-counting certain subspaces in the argument above, since there will be multiple sequences of vectors yielding the same subspace. The number of sequences yielding a fixed $d$-dimensional subspace can be counted in a similar manner as above and we get
			$(2^{d}-1)(2^{d}-2)(2^{d}-4)\cdots (2^{d}-2^{d-1})$.
			So the total number of subspaces is
			$$
			\frac{(2^n-1)(2^n-2)\cdots (2^n-2^{d-1})}{(2^{d}-1)(2^{d}-2)\cdots (2^{d}-2^{d-1})}\geq \frac{(2^n-2^{d-1})^{d}}{(2^{d}-1)^{d}} \geq 2^{\Omega((n-d)d)}.
			$$ 
		\end{proof}
		Combining this claim (with $d=n-\log k$) and $T=\Omega(k(\log |\V|)/n)$ gives $T=\Omega(k\log k)$.
	\end{proof}

\subsection{A potential direction to prove an improved sample complexity for Phase~1}
\label{sec: potential direction to improve phase 1}
In this section we give a potential direction to prove that in expectation $\Tilde{O}(k)$ Fourier-samples are sufficient for Phase~1 of our learning algorithm presented in Section~\ref{sec: Phase 1: Learning the Fourier span}. Recall
Phase~1 of our learning algorithm. Given a $k$-Fourier-sparse function $c$, Phase 1 starts by finding an $S \in \supp(\widehat{c})$ such that $S \neq 0^n$. Lemma~\ref{lem:phase1} implies that an expected number of $O(k)$ many Fourier-samples
are sufficient to sample such an $S$. Chang's lemma, a central result in additive combinatorics, gives tighter bound on the expected number of samples for this step. 
Chang's lemma upper bounds the dimension of the span of the ``large'' Fourier coefficients.
	
	
	\begin{lemma}[Chang's lemma]
	\label{lem:chang}
		Let $\alpha\in(0,1)$ and $\rho>0$. For every $f:\01^n\rightarrow \pmset{}$ that satisfies $\widehat{f}(0^n)=1-2\alpha$, we~have
		\begin{align}
		\label{eq:changinequality}
		\dim(\spann\{S:|\widehat{f}(S)|\geq \rho\alpha\})\leq \frac{2\log(1/\alpha)}{\rho^2}.
		\end{align}
	\end{lemma}
	
	Let us consider Chang's lemma for a $k$-Fourier-sparse Boolean function $c: \{0,1\}^n \to \{-1,1\}$ of Fourier dimension $r$ and let $\rho \in (0,1]$. In particular, consider the case $\rho\alpha=1/k$. In this case, since all elements of the Fourier support satisfy $|\widehat{c}(S)|\geq 1/k$ by Lemma~\ref{lemma:granularityofsparse}, the left-hand side of Eq.~\eqref{eq:changinequality} equals the Fourier dimension~$r$ of $c$. 
	Thus Chang's lemma gives
	$$
	r\leq 2\alpha^2 k^2 \log \rho k \leq 2\alpha^2 k^2 \log k,
	$$
	which implies
	\begin{align}
	    \sum_{S \neq 0^n} \widehat{c}(S)^2 = \Omega\left(\frac{\sqrt{r}} {k \sqrt{\log k}} \right). \label{eqn: old Chang lower bound}  
	\end{align}
	Thus an expected number of $O((k \sqrt{\log k})/\sqrt{r})$ many Fourier-samples
	are sufficient to obtain an $S\in \supp(\widehat{c})$ such that $S \neq 0^n$ in Phase~1. 
	This is already an improvement from what Lemma~\ref{lem:phase1} guaranteed.
	

    In this section we give an improvement of Chang's lemma for $k$-Fourier-sparse Boolean functions:
	\begin{theorem}\label{thm:improvedchang}
		Let $\alpha\in(0,1)$ and $k\geq 2$. For every $k$-Fourier-sparse $f:\01^n \to \pmset{}$ that satisfies $\widehat{f}(0^n)=1-2\alpha$ and $\Fdim(f)=r$, we~have
		$$
		\widehat{f}(0^n) \leq 1 - \frac{r}{k \log k}.
		$$
	\end{theorem}
	We remark that in a follow-up paper~\cite{chakraborty2020tight}, a subset of the authors gave a refinement of the theorem above.


		
	Before giving a proof of Theorem~\ref{thm:improvedchang}, let us first discuss how this theorem improves the analysis of Phase~1 of our learning algorithm. Theorem~\ref{thm:improvedchang} implies that for a $k$-Fourier-sparse Boolean function $c:\{0,1\}^n \to \{-1, 1\}$ of Fourier dimension $r$, 
	$$
	    \sum_{S: S\neq 0^n} \widehat{c}(S)^2 = \Omega(r/(k \log k)).
	$$
	This is a better lower bound on the Fourier weight of $c$ on the set $\{0,1\}^n \setminus \{0^n\}$ than that obtained from Chang's lemma (Equation~\ref{eqn: old Chang lower bound}). Thus an expected number of $O((k \log k)/r)$ many uniformly quantum samples is sufficient to obtain an $S \in \supp(\widehat{c})$ such that $S \neq 0^n$.
	
	We suspect that Theorem~\ref{thm:improvedchang} can in fact
		lead to an $\Tilde{O}(k)$ learning algorithm for Phase~1. Towards that end we make the following conjecture which can be viewed as a generalization of Theorem~\ref{thm:improvedchang}.
		
		\begin{conjecture}
		\label{conjecture : phase1} 
		    Let $n>0$ and $1\leq k\leq 2^n$. For every $k$-Fourier-sparse $f:\01^n \to \pmset{}$ with Fourier span~$\mathcal{V}$ and Fourier dimension $r$, the following holds: for every $r'>0$ and $\mathcal{S} \subset \mathcal{V}$ satisfying $\dim(\spann(\mathcal{S})) = r'$, we have
		$$
		\sum_{S\in \spann(\mathcal{S})} \widehat{f}(S)^2 \leq 1-\frac{r-r'}{k\log k}.
		$$
	    \end{conjecture}
	    If the above conjecture is true then it would imply an $\Tilde{O}(k)$ learning algorithm for Phase~1. Let $c : \{0,1\}^n \to \{-1,1\}$ be a $k$-Fourier-sparse function of Fourier dimension $r$. Assuming Conjecture~\ref{conjecture : phase1} to be true we have 
	    $$
	        \sum_{S\not\in \spann(\mathcal{S})}\widehat{c}(S)^2 \geq \frac{r-r'}{k\log k}.
	    $$
		So the expected number of samples to increase the dimension by $1$ is $\leq\frac{k\log k}{r-r'}$. Accordingly, the expected number of Fourier-samples needed to learn the whole Fourier span of $f$ is at most
		$$
		\sum_{i = 1}^{r} \frac{k\log k}{i}\leq O( k\log k\log r),
		$$
		where the final inequality used $\sum_{i=1}^r \frac{1}{i}=O(\log r)$. 		We now proceed to the proof of Theorem~\ref{thm:improvedchang}.
	   
	   \vspace{1em}

\subsubsection{Proof of Theorem~\ref{thm:improvedchang}}

		We first define the following notation. For $U\subseteq  [r]$, let $f^{(U)}$ be the function obtained by fixing the variables $\{x_i\}_{i\in U}$ in $f$ to $x_i = (1+\sign(\widehat{f}(i)))/2$ for all $i\in U$. Note that fixing variables cannot increase Fourier sparsity. For $i,j\in [r]$, define $f^{(i)}=f^{(\{i\})}$ and $f^{(ij)}=f^{(\{i,j\})}$. In this proof, for an invertible matrix $B\in \F_2^{n\times n}$, we will often treat its columns as a basis for the space~$\F_2^n$. Recall $f_B(x)=f((B^{-1})^Tx)$ from Definition~\ref{defn:fourierbasis}. We let $f_B^{(i)}$ be the function obtained by fixing $x_i = (1+\sign(\widehat{f}(i)))/2$ in the function $f_B$.
		

		The core idea in the proof of the theorem is the following structural lemma, which says that there is a particular $x_i$ that we can fix in the function $f_B$ without decreasing the Fourier dimension very much. 
		\begin{lemma}
			\label{lem:mainlemma}
			For every $k$-Fourier-sparse Boolean function $f:\01^n\rightarrow \pmset{}$ with $\Fdim(f)=r$, there exists an invertible matrix $B\in \F_2^{n\times n}$ and an index $i\in [r]$ such that $\Fdim(f_B^{(i)})\geq r-\log k$ and $\widehat{f_B}(j)\neq 0$ for all $j\in [r]$.
		\end{lemma}
		We defer the proof of the lemma to later and first conclude the proof of the theorem assuming the lemma. Consider the matrix~$B$ defined in Lemma~\ref{lem:mainlemma}. Using Lemma~\ref{lemma:basis} it follows that $f_B$ has only $r$ influential variables, so we can write $f_B:\01^r\rightarrow \pmset{}$, where $\widehat{f_B}(j)\neq 0$ for every $j\in [r]$. Also, $\widehat{f_B}(0^r)=\widehat{f}(0^n)=1-2\alpha$.  For convenience, we abuse notation and abbreviate~$f=f_B$. It remains to show that for every $f:\01^r\rightarrow \pmset{}$ with $\widehat{f}(j)\neq 0$ for all $j\in [r]$, we have $2\alpha=1-\widehat{f}(0^r) \geq r/(k\log k)$. We prove this by induction on~$r$.	

		\paragraph{Base case.} Let $r = 1$. Then $k=2$ (since $r\geq \log k$ and $k\geq 2$ by assumption). Note that the only Boolean functions with Fourier dimension $1$ and $|\supp(\widehat{f})|\leq 2$ are $\{\chi_j,-\chi_j\}$, where $\chi_j = (-1)^{x_j}$, for $j \in [n]$. In both these cases $1-\widehat{f}(0^r)=1$ and $r/(k\log k)=1/2$ (although the Fourier sparsity of $\chi_j$ is $1$, we are implicitly working with a concept class of $2$-sparse Boolean functions, hence $k=2$).
		
		\paragraph{Induction hypothesis.} Suppose that for all $p\in \{1,\ldots, r-1\}$ and  $k$-Fourier-sparse Boolean function $g:\01^p\rightarrow\pmset{}$ with $\Fdim(g)=p$ and $\widehat{g}(j)\neq 0$ for all $j\in [p]$, we have $1-\widehat{g}(0^p)\geq p/(k\log k)$.
		
		\paragraph{Induction step.} Let $i\in [r]$ be the index from Lemma~\ref{lem:mainlemma}. Note that $f^{(i)}$ is still $k$-Fourier-sparse and $\widehat{f^{(i)}}(0^{r-1})=1-2\alpha+|\widehat{f}(i)|$. Since $|\widehat{f}(i)|\geq 1/k$ (by Lemma~\ref{lemma:granularityofsparse}), we have
		\begin{align*}
		\widehat{f^{(i)}}(0^{r-1})\geq 1-2\alpha +1/k.
		\end{align*}
		Since $r-\log k\leq\Fdim(f^{(i)})\leq r-1$, we can use the induction hypothesis on the function~$f^{(i)}$ to conclude that
		\begin{align*}
		2\alpha\geq 1-\widehat{f^{(i)}}(0^{r-1})+\frac{1}{k}\geq \frac{r-\log k}{k\log k}+\frac{1}{k}=\frac{r}{k\log k}.
		\end{align*}
		This concludes the proof of the induction step and the theorem. We now prove Lemma~\ref{lem:mainlemma}.

		\begin{proof}[Proof of Lemma~\ref{lem:mainlemma}]
			In order to construct $B$ as in the lemma statement, we
			first make the following observation.
			
			\begin{observation}
				\label{obs:mainobservation}
				For every Boolean function $f: \{0,1\}^n\rightarrow \pmset{}$ with $\Fdim(f)=r$, there exists an invertible $B\in \F_2^{n\times n}$ such that:
				\begin{enumerate}
					\item The Fourier coefficient  $\widehat{f_{B}}(1)$ is non-zero.
					\item There exists a $t\in [r]$ such that,  for all $j\in \{2,\ldots,t\}$, we have $\Fdim(f_B^{(j)})\leq r-t$.
					\item The Fourier span of $f_{B}^{(1)}$ is spanned by $\{e_{t+1},\ldots,e_r\}$.
					\item For $\ell\in \{t+1,\ldots,r\}$, the Fourier coefficients $\widehat{f_{B}^{(1)}}(\ell)$ are non-zero.
				\end{enumerate}
			\end{observation}

			We defer the proof of this observation to the end. We proceed to prove the lemma assuming the observation. Note that Property~3 gives the following simple corollary:

			\begin{corollary}
				\label{cor:commutative}
				$f_B^{(1)}$ is a function of $x_{t+1}, \dots, x_r$ and independent of $x_2, \dots, x_t$ (and hence $f_B^{(1)}=f_B^{(i1)} = f_B^{(1i)}$ for every $i \in \{2,\dots, t\}$).
			\end{corollary}

			We now show that not only $f_B^{(1)}$, but all the functions $f_B^{(2)},\ldots,f_B^{(t)}$ are independent of $x_2, \dots, x_t$.

			\begin{claim}
				\label{claim:assign2tot}
				For all $i \in \{2, \dots, t\}$, $f_B^{(i)}$ is a function of $\{x_1, x_{t+1}, \dots, x_r\}$ and independent of $x_2, \dots, x_t$.
			\end{claim}
			
			\begin{proof}
				Without loss of generality, let $i = 2$. By Observation~\ref{obs:mainobservation} (property 4), the character functions $\chi_{t+1},\ldots,\chi_r$ are present in the Fourier expansion of $f_B^{(1)}$.
				We have $f_B^{(21)}=f_B^{(1)}$ by Corollary~\ref{cor:commutative}. Hence, for every $\ell\in \{t+1,\ldots,r\}$, at least one of the characters $\chi_\ell$ or $\chi_1\chi_{\ell}$ is present in the Fourier expansion of $f_B^{(2)}$. Let $y_\ell$ be $\chi_\ell$ or $\chi_1\chi_\ell$ (depending on which character function is present in the Fourier expansion of $f_B^{(2)}$). Note that the $r-t$ character functions $y_{t+1},\ldots,y_{r}$ are linearly independent. By Observation~\ref{obs:mainobservation} (Property $2$), we have $\Fdim(f_B^{(2)})\leq r-t$, which implies 
				$
				\Fspan({f_B^{(2)}}) \subseteq \spann \{y_{t+1}, \ldots, y_{r} \}
				$ and $f_B^{(2)}$ is independent of $\{x_2,\ldots,x_t\}$. The same argument shows that for every $i,k \in \{2,\ldots,t\}$, $f_B^{(i)}$ is independent of~$x_k$.
			\end{proof}
			\vspace{-8pt}
			
			\begin{claim}\label{claim:assignrest}
				There exists an assignment of $(x_1, x_{t+1}, \dots, x_r)$ to $(a_1,a_{t+1}\ldots,a_r)$ in $f_B$ such that the resulting function depends on \emph{all} variables $x_2,\ldots,x_t$.\footnote{Observe that in this assignment, we have $x_1=(1-\sign(\widehat{f}(1)))/2$. Otherwise, by assigning $x_1=(1+\sign(\widehat{f}(1)))/2$ in $f_B$, we would obtain the function $f_B^{(1)}$ which we know is independent of $\{x_2,\ldots,x_t\}$ by Corollary~\ref{cor:commutative}.}
			\end{claim}
			
			\begin{proof}
				Before proving the claim we first make the following observation. Let us consider an assignment of $(x_1, x_{t+1}, \dots, x_r)=z$ in $f_B$ and assume that the resulting function $f_{B,z}$ is independent of $x_i$  for some $i\in \{2,\ldots,t\}$. Let us assign $x_i  = (1+\sign(\widehat{f_B}(i)))/2$ in $f_{B,z}$ and call the resulting function $f_{B,z}^{(i)}$. Firstly, $f_{B,z}^{(i)}=f_{B,z}$ since $f_{B,z}$ was independent of $x_i$.  Secondly, observe that $f_{B,z}=f_{B,z}^{(i)}$ could have alternatively been obtained by \emph{first} fixing $x_i=(1+\sign(\widehat{f}(i)))/2$ in $f_{B}$ and then fixing $(x_1,x_{t+1}\ldots,x_r)=~z$. In this case, by Claim~\ref{claim:assign2tot}, after fixing $x_i$ in  $f_{B}$, $f_B^{(i)}$ is independent of $x_2,\ldots,x_t$ and after fixing $(x_1, x_{t+1}, \dots, x_r)=z$, $f_{B,z}$ is a constant. This in particular shows that if there exists a $z$ such that $f_{B,z}$ is independent of~$x_i$ for some $i\in \{2,\ldots,t\}$, then $f_{B,z}$ is also independent of $x_2,\ldots,x_t$. 			
				
				Towards a contradiction, suppose that for every assignment of $(x_1, x_{t+1}, \dots, x_r)=z$ to $f_B$, the resulting function $f_{B,z}$ is independent of $x_i$, for some $i\in \{2,\ldots,t\}$. Then by the argument in the previous paragraph, for every assignment $z$, $f_{B,z}$ is also independent of $x_k$ for \emph{every} $k\in \{2,\ldots,t\}$. This, however, contradicts the fact that $x_2,\ldots,x_t$ had non-zero influence on $f_B$ (since $B$ was chosen such that $\widehat{f_B}(j)\neq 0$ for every $j\in [r]$ in Lemma~\ref{lem:mainlemma}). This implies the existence of an assignment $(x_1,x_{t+1},\ldots,x_r)=(a_1,a_{t+1}\ldots,a_r)$, such that the resulting function depends on all the variables~$x_2,\ldots,x_t$.
			\end{proof}
\vspace{-8pt}

			We now argue that the assignment in Claim~\ref{claim:assignrest} results in a function which resembles the $\AND$ function on $x_2,\ldots,x_t$, and hence has Fourier sparsity~$2^{t-1}$.
			
			\begin{claim}\label{claim:assigand}
				Consider the assignment $(x_1,x_{t+1},\ldots,x_r)=(a_1,a_{t+1}\ldots,a_r)$ in $f_B$ as in Claim~\ref{claim:assignrest}, then the resulting function $g$ equals (up to possible negations of input and output bits) the $(t-1)$-bit $\AND$ function.
			\end{claim}
			
			\begin{proof}
				By Claim~\ref{claim:assignrest}, $g$ depends on all the variables $x_2,\ldots,x_t$. This dependence is such that if \emph{any one} of the variables $\{x_i:i\in \{2,\ldots,t\}\}$ is set to $x_i = (1+\sign(\widehat{f_B}(i)))/2$, then by Claim~\ref{claim:assign2tot} the resulting function $g^{(i)}$ is independent of $x_2,\ldots,x_t$. Hence, $g^{(i)}$ is some constant $b_i \in \pmset{}$ for every $i \in \{2, \dots, t\}$. Note that these $b_i$s are all the same bit~$b$, because first fixing $x_i$ (which collapses~$g$ to the constant~$b_i$) and then $x_j$ gives the same function as first fixing $x_j$ (which collapses~$g$ to~$b_j$) and then $x_i$. Additionally, by assigning $x_i = (1-\sign(\widehat{f_B}(i)))/2$ for every $i\in\{2,\ldots,t\}$ in $g$, the resulting function must evaluate to $1-b$ because $g$ is non-constant (it depends on $x_2,\ldots,x_t$). Therefore $g$ equals (up to possible negations of input and output bits) the $(t-1)$-bit $\AND$ function.	
			\end{proof}
			
			We now conclude the proof of Lemma~\ref{lem:mainlemma}. Let $f:\01^n \to \pmset{}$ be such that $\Fdim(f)=r$. Let~$B$ be as defined in Observation~\ref{obs:mainobservation}. Consider the assignment of $(x_{t+1},\ldots,x_r)=(a_{t+1},\ldots,a_r)$ to $f_B$ as in Claim~\ref{claim:assigand}, and call the resulting function $f_B'$. From Claim~\ref{claim:assigand}, observe that by setting $x_1=a_1$ in~$f_B'$, the resulting function is $g(x_2,\ldots,x_t)$ and by setting $x_1=1-a_1$ in~$f_B'$, the resulting function is a constant. Hence $f_B'$ can be written as
			\begin{align}
			\label{eq:rewritefB}
			f_B'(x_1,\ldots,x_t,a_{t+1},\ldots,a_r)=\frac{1-(-1)^{x_1+a_1}}{2}b_{a_1,a_{t+1},\ldots,a_r} +\frac{1+(-1)^{x_1+a_1}}{2} g(x_2,\ldots,x_t),
			\end{align}
			where $b_{a_1,a_{t+1},\ldots,a_r}\in \pmset{}$ (note that it is independent of $x_2,\ldots,x_t$ by Corollary~\ref{cor:commutative}). Since $g$ essentially equals the $(t-1)$-bit $\AND$ function (by Claim~\ref{claim:assigand}), $g$ has Fourier sparsity $2^{t-1}$ and $\widehat{g}(0^{t-1})=1-2^{-t+2}$. Hence the Fourier sparsity of $f_B'$ in Eq.~\eqref{eq:rewritefB} equals $2^{t}$. Since $f_B'$ was a restriction of $f_B$, the Fourier sparsity of $f_B'$ is at most $k$, hence $t\leq \log k$. This implies $\Fdim(f_B^{(1)} )=r - t \geq r - \log k$, concluding the~proof.
		\end{proof}
		
		It remains to prove Observation~\ref{obs:mainobservation}, which we do now.

		\begin{proof}[Proof of Observation~\ref{obs:mainobservation}] 
			Let $D\in \F_2^{n\times n}$ be an invertible matrix that maximizes $\Fdim(f_{D}^{(1)})$ subject to the constraint $\widehat{f_{D}}(1) \neq 0$. Suppose $\Fdim(f_D^{(1)})=r-t$. Let $d_1,\ldots,d_{r-t}$ be a basis of  $\Fspan(f_D^{(1)})$ such that $\widehat{f_D^{(1)}}(d_i) \neq 0$ for all $i \in [r-t]$. We  now construct an invertible $C\in \F_2^{n\times n}$ whose first $r$~columns form a basis for $\Fspan(f_D)$, as follows: let $c_1=e_1$, and for $i \in [r-t]$, fix $c_{t+i} = d_i$. Next, assign vectors $c_2,\ldots,c_t$ arbitrarily from $\Fspan(f_D)$, ensuring that $c_2,\ldots,c_t$ are linearly independent from $\{c_1,c_{t+1},\ldots,c_r\}$. We then extend to a basis $\{c_1,\ldots,c_n\}$ arbitrarily. Define $C$ as $C=[c_1,\ldots,c_n]$ (where the $c_i$s are column vectors). Finally, define our desired matrix $B$ as the product $B = DC$. We now verify the properties of $B$.

			\textbf{Property 1:} Using Lemma~\ref{lemma:basis} we have 
			$$
			\widehat{f_{DC}}(1) = \widehat{f_{D}}(C e_1) = \widehat{f_{D}}(c_1) = \widehat{f_{D}}(1)\neq 0,
			$$
			where the third equality used $c_1=e_1$, and $\widehat{f_{D}}(1)\neq 0$ follows from the definition of~$D$.
			
			We next prove the following fact, which we use to verify the remaining three properties.
			
			\begin{fact}
				\label{fact:factaboutDB}
				Let $C,D$ be invertible matrices as defined above. For every $i\in [t]$, let $(f_D^{(i)})_C$ be the function obtained after applying the invertible transformation $C$ to $f_D^{(i)}$ and $(f_{DC})^{(i)}$ be the function obtained after fixing $x_i$ to $(1+\sign(\widehat{f_{DC}}(i)))/2$ in $f_{DC}$. Then $ (f_{DC})^{(i)}=(f_D^{(i)})_C $.
			\end{fact}

			\textbf{Property 2:} Fact~\ref{fact:factaboutDB} implies that  $\Fdim((f_{DC})^{(i)})=\Fdim((f_D^{(i)})_C)$. Since $C$ is invertible, $\Fdim((f_D^{(i)})_C)=\Fdim(f_D^{(i)})$. From the choice of $D$, observe that for all $i\in \{2,\ldots,t\}$, 
			$$
			\Fdim(f_B^{(i)})=\Fdim(f_{DC}^{(i)})=\Fdim((f_D^{(i)})_C)=\Fdim(f_{D}^{(i)})\leq\Fdim(f_{D}^{(1)})=r-t,
			$$
			where the inequality follows by definition of~$D$.

			\textbf{Property 3:} Note that $\Fspan(f_D^{(1)})$ is contained in $\spann\{ d_1, \dots, d_{r-t} \}$ by construction. By making the invertible transformation by $C$, observe that $\Fspan((f_D^{(1)})_C)\subseteq \spann\{ e_{t+1}, \dots, e_{r} \}$ (since for all $i\in [r-t]$, we defined $c_{t+i}=d_i$). Property~3 follows because $(f_D^{(1)})_C=f_{DC}^{(1)}=f_{B}^{(1)}$ by Fact~\ref{fact:factaboutDB}.

			\textbf{Property 4:}  Using Fact~\ref{fact:factaboutDB}, for every $\ell\in \{t+1,\ldots,r\}$, we have
			$$
			\widehat{(f_{B})^{(1)}}(\ell)=\widehat{(f_{DC})^{(1)}}(\ell) = \widehat{(f_D^{(1)})_C}(\ell)=\widehat{f_{D}^{(1)}}(c_\ell).
			$$
			Since $c_{\ell}=d_{\ell-t}$, we have $\widehat{f_{D}^{(1)}}(c_\ell) = \widehat{f_{D}^{(1)}}(d_{\ell-t})$ and $\widehat{f_{D}^{(1)}}(d_1),\ldots,\widehat{f_{D}^{(1)}}(d_{r-t}) \neq 0$ by definition of $d_i$, hence the property follows.
			
	\begin{proof}[Proof of Fact~\ref{fact:factaboutDB}]
				Let $f_D = g$. We want to show that $(g^{(i)} )_{C} = (g_{C})^{(i)}$. 
				For simplicity fix $i=1$; the same proof works for every $i\in [t]$. Then,
				\[
				(g^{(1)} )(x) = \sum_{S \in \{0\}\times \01^{n-1}} (\widehat{g}(S) + \widehat{g}(S \oplus e_1)) \chi_S(x).
				\]
				On transforming $g^{(1)}$ using the basis $C$ we have:
				
				\begin{align}
				\label{eq:g1bfourier}
				\begin{aligned}
				(g^{(1)} )_{C}(x) = \sum_{S \in \{0\}\times \01^{n-1}} (\widehat{g}(CS) + \widehat{g}(C(S \oplus e_1)) \chi_S(x). 
				\end{aligned}
				\end{align}
				Consider the function $g_C$. The Fourier expansion of $g_C$ is $
				g_C(y) = \sum_{S \in \01^n} \widehat{g}(CS) \chi_S(y)$ and the Fourier expansion of the  $(g_C)^{(1)}$ can be written as
				\begin{align}
				\label{eq:gb1fourier}
				g_{C}^{(1)} (y) &=  \sum_{S \in \{0\}\times \01^{n-1}} (\widehat{g}(CS) + \widehat{g}(CS \oplus Ce_1)) \chi_S(y).
				\end{align}
				Using Eq.~\eqref{eq:g1bfourier},~\eqref{eq:gb1fourier}, we conclude that $(g^{(1)} )_{C} = (g_{C})^{(1)}$, concluding the proof of the fact.
			\end{proof}
			This concludes the proof of the observation.
		\end{proof}
		This concludes the proof of the theorem.

	\section{Quantum vs classical membership queries}\label{sec:QvsD}
	In this section we assume we can access the target function using membership queries rather than examples.
	Our goal is to simulate quantum exact learners for a concept class $\Cc$ by classical exact learners, without using many more membership queries. A key tool here will be the (``nonnegative'' or ``positive-weights'') adversary method. This was introduced by Ambainis~\cite{ambainis:lowerboundsj}; here we will use the formulation of Barnum et al.~\cite{bss:semidef}, which is called the ``spectral adversary'' in the survey~\cite{spalek&szegedy:adversary}. 
	
	Let $\Cc\subseteq\01^N$ be a set of strings. If $N=2^n$ then we may view such a string $c\in\Cc$ as (the truth-table of) an $n$-bit Boolean function, but in this section we do not need the additional structure of functions on the Boolean cube and may consider any positive integer~$N$.
	Suppose we want to identify an unknown $c\in \Cc$ with success probability at least $2/3$ (i.e., we want to compute the identity function on $\Cc$). The required number of quantum queries to $c$ can be lower bounded as follows. Let $\Gamma$ be a $|\Cc|\times|\Cc|$ matrix with real, nonnegative entries and 0s on the diagonal (called an ``adversary matrix''). Let $D_i$ denote the $|\Cc|\times|\Cc|$ 0/1-matrix whose $(c,c')$-entry is $[c_i\neq c'_i]$.\footnote{The bracket-notation $[P]$ denotes the truth-value of proposition $P$.} 
	Then it is known that at least (a constant factor times) $\norm{\Gamma}/\max_{i\in[N]}\norm{\Gamma\circ D_i}$ quantum queries are needed, where $\norm{\cdot}$ denotes operator norm (largest singular value) and `$\circ$' denotes entrywise product of matrices. Let 
	$$
	\ADV(\Cc)=\max_{\Gamma\geq 0}\frac{\norm{\Gamma}}{\max_{i\in[N]}\norm{\Gamma\circ D_i}}
	$$ 
	denote the best-possible lower bound on $Q(\Cc)$ that can be achieved this way.
	
	The key to our classical simulation is the next lemma. It shows that if $Q(\Cc)$ (and hence~$\ADV(\Cc)$) is small, then there is a query that splits the concept class in a ``mildly balanced'' way.
	
	\begin{lemma}\label{lem:goodi}
		Let $\Cc\subseteq\01^N$ be a concept class and 
		$$
		    \ADV(\Cc)= \max_{\Gamma\geq 0}\frac{\norm{\Gamma}}{\max_{i\in[N]} \norm{\Gamma\circ D_i}}
		$$
		be the nonnegative adversary bound for the exact learning problem corresponding to $\Cc$. Let $\mu$ be a distribution on~$\Cc$ such that $\max_{c\in\Cc}\mu(c)\leq 5/6$.
		Then there exists an $i\in[N]$ such that 
		$$
		\min(\mu(C_i=0),\mu(C_i=1))\geq \frac{1}{36\ADV(\Cc)^2}.
		$$
	\end{lemma}
	
	\begin{proof}
		Define unit vector $v\in\mathbb{R}_+^{|\Cc|}$ by $v_c=\sqrt{\mu(c)}$, and adversary matrix 
		$$
		\Gamma=vv^* - \diag(\mu),
		$$
		where $\diag(\mu)$ is the diagonal matrix that has the entries of $\mu$ on its diagonal.
		This~$\Gamma$ is a nonnegative matrix with 0 diagonal (and hence a valid adversary matrix for the exact learning problem), and $\norm{\Gamma}\geq\norm{vv^*}-\norm{\diag(\mu)}\geq 1-5/6=1/6$. 
		Abbreviate $A=\ADV(\Cc)$.
		By definition of $A$, we have for this particular $\Gamma$
		$$
		A\geq \frac{\norm{\Gamma}}{\max_i\norm{\Gamma\circ D_i}}\geq \frac{1}{6\max_i\norm{\Gamma\circ D_i}},
		$$
		hence there exists an $i\in[N]$ such that $\norm{\Gamma\circ D_i}\geq\frac{1}{6A}$.
		We can write $v=\left(\begin{array}{c}v_0\\ v_1\end{array}\right)$ where the entries of~$v_0$ are the ones corresponding to $C$s where $C_i=0$, and the entries of $v_1$ are the ones where $C_i=1$. Then
		$$
		\Gamma=\left(\begin{array}{cc}
		v_0v_0^* & v_0v_1^*\\
		v_1v_0^* & v_1v_1^*
		\end{array}
		\right)-\diag(\mu)
		\mbox{~~~~and~~~~}
		\Gamma\circ D_i=\left(\begin{array}{cc}
		0 & v_0v_1^*\\
		v_1v_0^* & 0
		\end{array}
		\right).
		$$ 
		It is easy to see that $\norm{\Gamma\circ D_i}=\norm{v_0}\cdot\norm{v_1}$.
		Hence
		$$
		\frac{1}{36A^2}\leq \norm{\Gamma\circ D_i}^2=\norm{v_0}^2\norm{v_1}^2 = \mu(C_i=0)\mu(C_i=1)\leq \min(\mu(C_i=0),\mu(C_i=1)),
		$$ 
		where the last inequality used $\max(\mu(C_i=0),\mu(C_i=1))\leq 1$.
	\end{proof}
	
	Note that if we query the index $i$ given by this lemma and remove from $\Cc$ the strings that are inconsistent with the query outcome, then we reduce the size of $\Cc$ by a factor $\leq 1-\Omega(1/\ADV(\Cc)^2)$. Repeating this $O(\ADV(\Cc)^2\log|\Cc|)$ times would reduce the size of $\Cc$ to~1, completing the learning task.
	However, we will see below that analyzing the same approach in terms of entropy gives a somewhat better upper bound on the number of queries.
	
	\begin{theorem}\label{th:DvsADV}
		Let $\Cc\subseteq\01^N$ be a concept class and 
		$$
		    \ADV(\Cc)= \max_{\Gamma\geq 0}\frac{\norm{\Gamma}}{\max_{i\in[N]} \norm{\Gamma\circ D_i}}
		$$
		be the nonnegative adversary bound for the exact learning problem corresponding to $\Cc$.
		Then there exists a classical learner for $\Cc$ using $\displaystyle O\left(\frac{\ADV(\Cc)^2}{\log \ADV(\Cc)}\log|\Cc|\right)$ membership queries that identifies the target concept with probability $\geq 2/3$.
	\end{theorem}
	
	\begin{proof}
		Fix an arbitrary distribution $\mu$ on $\Cc$.
		We will construct a deterministic classical learner for $\Cc$ with success probability $\geq 2/3$ under~$\mu$. Since we can do this for every $\mu$, the ``Yao principle''~\cite{yao:unified} then implies the existence of a randomized learner that has success probability $\geq 2/3$ for every $c\in\Cc$.
		
		Consider the following algorithm, whose input is an $N$-bit random variable $C\sim\mu$:
		\begin{enumerate}
			\item Choose an $i$ that maximizes $H(C_i)$ and query that $i$.\footnote{Querying this~$i$ will give a fairly ``balanced'' reduction of the size of $\Cc$ irrespective of the outcome of the query. If there are several maximizing $i$s, then choose the smallest $i$ to make the algorithm deterministic.}
			\item Update $\Cc$ and $\mu$ by restricting to the concepts that are consistent with the query outcome.
			\item Goto~1.
		\end{enumerate}
		The queried indices are themselves random variables, and we denote them by $I_1,I_2,\ldots$. We can think of $t$ steps of this algorithm as generating a binary tree of depth~$t$, where the different paths correspond to the different queries made and their binary outcomes.
		
		Let $P_t$ be the probability that, after $t$ queries, our algorithm has reduced $\mu$ to a distribution that has weight $\geq 5/6$ on one particular $c$:
		\begin{align*}
		    & P_t = \\
		    & \sum_{
		    \substack{i_1,\ldots,i_t\in[N] \\ b\in\01^t}}
		    \Pr[I_1=i_1,\ldots,I_t=i_t,C_{i_1}\ldots C_{i_t}=b]\cdot \left[\exists c\in\Cc~s.t.~\mu(c\mid C_{i_1}\ldots C_{i_t}=b)\geq \frac{5}{6}\right].
		\end{align*}
		Because restricting $\mu$ to a subset $\Cc'\subseteq\Cc$ cannot decrease probabilities of individual~$c\in\Cc'$, this probability~$P_t$ is non-decreasing in~$t$. Because $N$ queries give us the target concept completely, we have~$P_N=1$. 
		Let $T$ be the smallest integer~$t$ for which $P_t\geq 5/6$. We will run our algorithm for $T$ queries, and then output the $c$ with highest probability under the restricted version of $\mu$ we now have.  With $\mu$-probability at least $5/6$, that $c$ will have probability at least $5/6$ (under $\mu$ conditioned on the query-results).
		The overall error probability under $\mu$ is therefore $\leq 1/6+1/6=1/3$.
		
		It remains to upper bound~$T$. To this end, define the following ``energy function'' in terms of conditional entropy:
		\begin{align*}
		    E_t 
		    &= H(C\mid C_{I_1},\ldots,C_{I_t}) \\
		    &=
		    \sum_{
		    \substack{i_1,\ldots,i_t\in[N] \\ b\in\01^t}}
		    \Pr[I_1=i_1,\ldots,I_t=i_t,C_{i_1}\ldots C_{i_t}=b]\cdot H(C\mid C_{i_1}\ldots C_{i_t}=b).
		\end{align*}
		Because conditioning on a random variable cannot increase entropy, $E_t$ is non-increasing in $t$. We will show below that as long as $P_t<5/6$, the energy shrinks significantly with each new query.
		
		Let $C_{i_1}\ldots C_{i_t}=b$ be such that there is no $c\in\Cc$~s.t.~$\mu(c\mid C_{i_1}\ldots C_{i_t}=b)\geq 5/6$ (note that this event happens in our algorithm with $\mu$-probability $1-P_t$). Let $\mu'$ be $\mu$ restricted to the class $\Cc'$ of concepts $c$ where $c_{i_1}\ldots c_{i_t}=b$. The nonnegative adversary bound for this restricted concept class is $A'=\ADV(\Cc')\leq\ADV(\Cc)=A$. 
		Applying Lemma~\ref{lem:goodi} to $\mu'$, there is an $i_{t+1}\in[N]$ with $p:=\min(\mu'(C_{i_{t+1}}=0),\mu'(C_{i_{t+1}}=1))\geq \frac{1}{36A'^2}\geq \frac{1}{36 A^2}$. Note that $H(p)\geq \Omega(\log(A)/A^2)$. Hence 
		\begin{align*}
		H(C\mid C_{i_1}\ldots C_{i_t}=b)- H(C\mid C_{i_1}\ldots C_{i_t}=b,C_{i_{t+1}})
		&= H(C_{i_{t+1}}\mid C_{i_1}\ldots C_{i_t}=b) \\
		&\geq \Omega(\log(A)/A^2).
		\end{align*}
		This implies $E_t-E_{t+1}\geq (1-P_t)\cdot\Omega(\log(A)/A^2)$. 
		In particular, as long as $P_t<5/6$, the $(t+1)$st query shrinks $E_t$ by at least $\frac{1}{6}\Omega(\log(A)/A^2)=\Omega(\log(A)/A^2)$. Since $E_0=H(C)\leq \log|\Cc|$ and $E_t$ cannot shrink below~0, there can be at most $\displaystyle O\left(\frac{A^2}{\log A}\log|\Cc|\right)$ queries before $P_t$ grows to $\geq 5/6$.
	\end{proof}
	
	Since $\ADV(\Cc)$ lower bounds $Q(\Cc)$, Theorem~\ref{th:DvsADV} implies the bound $$R(\Cc)\leq O\left(\frac{Q(\Cc)^2}{\log Q(\Cc)}\log|\Cc|\right)$$ claimed in our introduction.
	Note that this bound is tight up to a constant factor for the class of $N$-bit point functions, where $Q(\Cc) = \Theta(\sqrt{N})$, $|\Cc|=N$, and $R(\Cc)=\Theta(N)$ classical queries are necessary and sufficient.

	\section{Future work}
	
	Neither of our two results is tight. As directions for future work, let us state two conjectures, one for each model:
	\begin{itemize}
		\item $k$-Fourier-sparse functions can be learned from $O(k\cdot\polylog(k))$ uniform quantum~examples.
		
		\item For all concept classes $\Cc$ of Boolean-valued functions on a domain of size~$N$ we have:\\ $R(\Cc)=O(Q(\Cc)^2 + Q(\Cc)\log N)$.
	\end{itemize}
	
\paragraph{Acknowledgements.} We thank Swagato Sanyal for pointing out an error in a previous version of this paper.

\bibliographystyle{plainnat}
\bibliography{qcs}

\begin{thebibliography}{35}
\providecommand{\natexlab}[1]{#1}
\providecommand{\url}[1]{\texttt{#1}}
\expandafter\ifx\csname urlstyle\endcsname\relax
  \providecommand{\doi}[1]{doi: #1}\else
  \providecommand{\doi}{doi: \begingroup \urlstyle{rm}\Url}\fi

\bibitem[Adcock et~al.(2015)Adcock, Allen, Day, Frick, Hinchliff, Johnson,
  Morley-Short, Pallister, Price, and Stanisic]{adcockea:qml}
J.~Adcock, E.~Allen, M.~Day, S.~Frick, J.~Hinchliff, M.~Johnson,
  S.~Morley-Short, S.~Pallister, A.~Price, and S.~Stanisic.
\newblock Advances in quantum machine learning, 2015.
\newblock URL \url{https://arxiv.org/abs/1512.02900}.

\bibitem[Ambainis(2002)]{ambainis:lowerboundsj}
A.~Ambainis.
\newblock Quantum lower bounds by quantum arguments.
\newblock \emph{Journal of Computer and System Sciences}, 64\penalty0
  (4):\penalty0 750--767, 2002.
\newblock \doi{10.1006/jcss.2002.1826}.
\newblock Earlier version in STOC'00.

\bibitem[Arunachalam and Wolf(2017)]{arunachalam:quantumlearningsurvey}
S.~Arunachalam and R.~{de} Wolf.
\newblock Guest column: {A} survey of quantum learning theory.
\newblock \emph{{SIGACT} News}, 48\penalty0 (2):\penalty0 41--67, 2017.
\newblock \doi{10.1145/3106700.3106710}.
\newblock arXiv:1701.06806.

\bibitem[Arunachalam and Wolf(2018)]{arunachalam:optimalpaclearning}
S.~Arunachalam and R.~{de} Wolf.
\newblock Optimal quantum sample complexity of learning algorithms.
\newblock \emph{Journal of Machine Learning Research}, 19, 2018.
\newblock URL \url{http://jmlr.org/papers/v19/18-195.html}.
\newblock Earlier version in CCC'17.

\bibitem[{At\i c\i} and Servedio(2009)]{atici&servedio:testing}
A.~{At\i c\i} and R.~Servedio.
\newblock Quantum algorithms for learning and testing juntas.
\newblock \emph{Quantum Information Processing}, 6\penalty0 (5):\penalty0
  323--348, 2009.
\newblock \doi{10.1007/s11128-007-0061-6}.

\bibitem[Barnum et~al.(2003)Barnum, Saks, and Szegedy]{bss:semidef}
H.~Barnum, M.~Saks, and M.~Szegedy.
\newblock Quantum query complexity and semi-definite programming.
\newblock In \emph{Proceedings of 18th IEEE Conference on Computational
  Complexity}, pages 179--193, 2003.
\newblock \doi{10.1109/CCC.2003.1214419}.

\bibitem[Bernstein and Vazirani(1997)]{bernstein&vazirani:qcomplexity}
E.~Bernstein and U.~Vazirani.
\newblock Quantum complexity theory.
\newblock \emph{SIAM Journal on Computing}, 26\penalty0 (5):\penalty0
  1411--1473, 1997.
\newblock \doi{10.1137/S0097539796300921}.
\newblock Earlier version in STOC'93.

\bibitem[Biamonte et~al.(2017)Biamonte, Wittek, Pancotti, Rebentrost, Wiebe,
  and Lloyd]{biamonteea:qml}
J.~Biamonte, P.~Wittek, N.~Pancotti, P.~Rebentrost, N.~Wiebe, and S.~Lloyd.
\newblock Quantum machine learning.
\newblock \emph{Nature}, 549\penalty0 (7671), 2017.
\newblock \doi{10.1038/nature23474}.

\bibitem[Bourgain(2014)]{bourgain:RIP}
J.~Bourgain.
\newblock An improved estimate in the restricted isometry problem.
\newblock In \emph{Geometric Aspects of Functional Analysis}, volume 2116 of
  \emph{Lecture Notes in Mathematics}, pages 65--70. Springer, 2014.
\newblock \doi{10.1007/978-3-319-09477-9_5}.

\bibitem[Bshouty and Jackson(1999)]{bshouty:quantumpac}
N.~H. Bshouty and J.~C. Jackson.
\newblock Learning {DNF} over the uniform distribution using a quantum example
  oracle.
\newblock \emph{SIAM Journal on Computing}, 28\penalty0 (3):\penalty0
  1136–--1153, 1999.
\newblock \doi{10.1145/225298.225312}.
\newblock Earlier version in COLT'95.

\bibitem[Cand{\'{e}}s and Tao(2006)]{candesandtao:signalrecovery}
E.~J. Cand{\'{e}}s and T.~Tao.
\newblock Near-optimal signal recovery from random projections: Universal
  encoding strategies?
\newblock \emph{{IEEE} Transactions on Information Theory}, 52\penalty0
  (12):\penalty0 5406--5425, 2006.
\newblock \doi{10.1109/TIT.2006.885507}.

\bibitem[Chakraborty et~al.(2020)Chakraborty, Mande, Mittal, Molli, Paraashar,
  and Sanyal]{chakraborty2020tight}
Sourav Chakraborty, Nikhil~S Mande, Rajat Mittal, Tulasimohan Molli, Manaswi
  Paraashar, and Swagato Sanyal.
\newblock Tight {C}hang's-lemma-type bounds for {B}oolean functions, 2020.
\newblock URL \url{https://arxiv.org/abs/2012.02335}.

\bibitem[Chang(2002)]{chang:inequality}
M.~C. Chang.
\newblock A polynomial bound in {F}reiman's theorem.
\newblock \emph{Duke Mathematics Journal}, 113\penalty0 (3):\penalty0 399--419,
  2002.
\newblock \doi{10.1215/S0012-7094-02-11331-3}.

\bibitem[Cheraghchi et~al.(2013)Cheraghchi, Guruswami, and
  Velingker]{cheraghchi:RIPlistdecoding}
M.~Cheraghchi, V.~Guruswami, and A.~Velingker.
\newblock Restricted isometry of {F}ourier matrices and list decodability of
  random linear codes.
\newblock \emph{SIAM Journal on Computing}, 42\penalty0 (5):\penalty0
  1888--1914, 2013.
\newblock \doi{10.1137/120896773}.

\bibitem[Cover and Thomas(1991)]{cover&thomas:infoth}
T.~M. Cover and J.~A. Thomas.
\newblock \emph{Elements of Information Theory}.
\newblock Wiley, 1991.
\newblock \doi{10.1002/047174882X}.

\bibitem[Dunjko and Briegel(2018)]{briegel&dunjko:qml}
V.~Dunjko and H.~J. Briegel.
\newblock Machine learning \& artificial intelligence in the quantum domain: a
  review of recent progress.
\newblock \emph{Reports on Progress in Physics}, 81\penalty0 (7):\penalty0
  074001, 2018.
\newblock \doi{doi:10.1088/1361-6633/aab406}.

\bibitem[Gopalan et~al.(2011)Gopalan, O'Donnell, Servedio, Shpilka, and
  Wimmer]{gopalan:fouriersparsity}
P.~Gopalan, R.~O'Donnell, R.~A. Servedio, A.~Shpilka, and K.~Wimmer.
\newblock Testing {F}ourier dimensionality and sparsity.
\newblock \emph{SIAM Journal on Computing}, 40\penalty0 (4):\penalty0
  1075--1100, 2011.
\newblock \doi{10.1137/100785429}.
\newblock Earlier version in ICALP'09.

\bibitem[Grover(1996)]{grover:search}
L.~K. Grover.
\newblock A fast quantum mechanical algorithm for database search.
\newblock In \emph{Proceedings of 28th ACM STOC}, pages 212--219, 1996.
\newblock \doi{10.1145/237814.237866}.

\bibitem[Harrow et~al.(2009)Harrow, Hassidim, and Lloyd]{hhl:lineq}
A.~Harrow, A.~Hassidim, and S.~Lloyd.
\newblock Quantum algorithm for solving linear systems of equations.
\newblock \emph{Physical Review Letters}, 103\penalty0 (15):\penalty0 150502,
  2009.
\newblock \doi{10.1103/PhysRevLett.103.150502}.

\bibitem[Hassanieh et~al.(2012)Hassanieh, Indyk, Katabi, and
  Price]{hassanieh:sparserecovery}
H.~Hassanieh, P.~Indyk, D.~Katabi, and E.~Price.
\newblock Nearly optimal sparse {F}ourier transform.
\newblock In \emph{Proceedings of 44th ACM STOC}, pages 563--578, 2012.
\newblock \doi{10.1145/2213977.2214029}.

\bibitem[Haviv and Regev(2016)]{haviv:listdecoding}
I.~Haviv and O.~Regev.
\newblock The list-decoding size of {F}ourier-sparse {B}oolean functions.
\newblock \emph{ACM Transactions on Computation Theory}, 8\penalty0
  (3):\penalty0 10:1--10:14, 2016.
\newblock \doi{10.1145/2898439}.
\newblock Earlier version in CCC'15.

\bibitem[Indyk and Kapralov(2014)]{indyk:sampleoptimal}
P.~Indyk and M.~Kapralov.
\newblock Sample-optimal {F}ourier sampling in any constant dimension.
\newblock In \emph{Proceedings of 55th IEEE FOCS}, pages 514--523, 2014.
\newblock \doi{10.1109/FOCS.2014.61}.

\bibitem[Mossel et~al.(2004)Mossel, {O'Donnell}, and
  Servedio]{mos:learningjuntas}
E.~Mossel, R.~{O'Donnell}, and R.~Servedio.
\newblock Learning functions of {$k$} relevant variables.
\newblock \emph{Journal of Computer and System Sciences}, 69\penalty0
  (3):\penalty0 421--434, 2004.
\newblock \doi{10.1016/j.jcss.2004.04.002}.
\newblock Earlier version in STOC'03.

\bibitem[Nielsen and Chuang(2010)]{nielsen&chuang:qc}
Michael~A. Nielsen and Isaac~L. Chuang.
\newblock \emph{Quantum Computation and Quantum Information: 10th Anniversary
  Edition}.
\newblock Cambridge University Press, 2010.
\newblock \doi{10.1017/CBO9780511976667}.

\bibitem[{O'Donnell}(2014)]{odonnell:analysis}
R.~{O'Donnell}.
\newblock \emph{Analysis of Boolean Functions}.
\newblock Cambridge University Press, 2014.
\newblock \doi{10.1017/CBO9781139814782}.

\bibitem[Rudelson and Vershynin(2008)]{rudelsyn:sparsereconstruction}
M.~Rudelson and R.~Vershynin.
\newblock On sparse reconstruction from {F}ourier and {G}aussian measurements.
\newblock \emph{Communications on Pure and Applied Mathematics}, 61\penalty0
  (8):\penalty0 1025--1045, 2008.
\newblock \doi{10.1002/cpa.20227}.

\bibitem[Sanyal(2019)]{sanyal:fourierdim}
Swagato Sanyal.
\newblock Fourier sparsity and dimension.
\newblock volume~15, pages 1--13. Theory of Computing, 2019.
\newblock \doi{10.4086/toc.2019.v015a011}.

\bibitem[Schuld et~al.(2015)Schuld, Sinayskiy, and
  Petruccione]{schuldea:introqml}
M.~Schuld, I.~Sinayskiy, and F.~Petruccione.
\newblock An introduction to quantum machine learning.
\newblock \emph{Contemporary Physics}, 56\penalty0 (2):\penalty0 172--185,
  2015.
\newblock \doi{10.1080/00107514.2014.964942}.

\bibitem[Servedio and
  Gortler(2004)]{servedio&gortler:equivalencequantumclassical}
R.~Servedio and S.~Gortler.
\newblock Equivalences and separations between quantum and classical
  learnability.
\newblock \emph{SIAM Journal on Computing}, 33\penalty0 (5):\penalty0
  1067--1092, 2004.
\newblock \doi{10.1137/S0097539704412910}.
\newblock Combines earlier papers from ICALP'01 and CCC'01.

\bibitem[{\v{S}}palek and Szegedy(2005)]{spalek&szegedy:adversary}
R.~{\v{S}}palek and M.~Szegedy.
\newblock All quantum adversary methods are equivalent.
\newblock In \emph{Proceedings of 32nd ICALP}, volume 3580 of \emph{Lecture
  Notes in Computer Science}, pages 1299--1311, 2005.
\newblock \doi{10.1007/11523468_105}.

\bibitem[Valiant(1984)]{valiant:paclearning}
L.~Valiant.
\newblock A theory of the learnable.
\newblock \emph{Communications of the ACM}, 27\penalty0 (11):\penalty0
  1134--1142, 1984.
\newblock \doi{10.1145/1968.1972}.

\bibitem[Verbeurgt(1990)]{verbeurgt:learningdnf}
K.~A. Verbeurgt.
\newblock Learning {DNF} under the uniform distribution in quasi-polynomial
  time.
\newblock In \emph{Proceedings of 3rd Annual Workshop on Computational Learning
  Theory (COLT'90)}, pages 314--326, 1990.
\newblock URL \url{https://dl.acm.org/doi/10.5555/92571.92659}.

\bibitem[Wittek(2014)]{wittek:qml}
P.~Wittek.
\newblock \emph{Quantum Machine Learning: What Quantum Computing Means to Data
  Mining}.
\newblock Elsevier, 2014.
\newblock \doi{10.1016/C2013-0-19170-2}.

\bibitem[Wolf(2008)]{wolf:fouriersurvey}
R.~{de} Wolf.
\newblock A brief introduction to {F}ourier analysis on the {B}oolean cube.
\newblock \emph{Theory of Computing}, 2008.
\newblock \doi{10.4086/toc.gs.2008.001}.
\newblock ToC Library, Graduate Surveys 1.

\bibitem[Yao(1977)]{yao:unified}
A.~C-C. Yao.
\newblock Probabilistic computations: Toward a unified measure of complexity.
\newblock In \emph{Proceedings of 18th IEEE FOCS}, pages 222--227, 1977.
\newblock \doi{10.1109/SFCS.1977.24}.

\end{thebibliography}

\end{document}